\documentclass[12pt,a4paper,twoside]{article}
\usepackage{hyperref}
\usepackage{upgreek}
\usepackage{textgreek}
\usepackage{amssymb,amsmath,amsthm}
\usepackage{cite}
\usepackage{tikz}
\usepackage[utf8]{inputenc}
\usetikzlibrary{arrows,backgrounds,decorations.pathmorphing,decorations.pathreplacing,positioning,fit}
\usepackage{enumerate}
\usepackage{multicol}
\usepackage{accents}
\usepackage{stackengine}
\usepackage[conditional,light,first,bottomafter]{draftcopy}
\draftcopyName{DRAFT\space\today}{130}
\draftcopySetScale{65}
\usepackage{lmodern}
\usepackage{microtype}
\usepackage{graphicx}
\usepackage{subcaption}
\usepackage{dsfont}
\usepackage{colonequals}
\usepackage{fullpage}
\usepackage{setspace}
\makeatletter
\renewcommand{\section}{\@startsection%
{section}%
{1}%
{0em}%
{1.7em}%
{1.2em}%
{\normalfont\large\centering\bfseries}}
\renewcommand{\@seccntformat}[1]%
{\csname the#1\endcsname.\hspace{0.5em}}
\makeatother

\numberwithin{equation}{section}
\newtheorem{theorem}{Theorem}[section]
\newtheorem{proposition}[theorem]{Proposition}
\newtheorem{lemma}[theorem]{Lemma}
\newtheorem{corollary}[theorem]{Corollary}
\newtheorem*{conjecture}{Conjecture}
\theoremstyle{definition}
\newtheorem{definition}[theorem]{Definition}
\newtheorem{remark}[theorem]{Remark}

\newtheorem*{acknowledgments}{Acknowledgments}

\newcommand{\abs}[1]{\left|#1\right|}
\newcommand\rEu[1]{_{\upharpoonright_{#1}}}
\newcommand\dS{\displaystyle}
\newcommand{\llb}{\left\lbrace}
\newcommand{\rrb}{\right\rbrace}
\newcommand\N{{\mathbb N}}
\newcommand\C{{\mathbb C}}
\renewcommand{\H}{\mathcal{H}}
\newcommand{\ceq}{\colonequals}
\newcommand\cc[1]{\overline {#1}}
\newcommand\vb[1]{\langle {#1}|}
\newcommand\vk[1]{|{#1}\rangle}
\newcommand\vba[1]{\left\langle {#1} \right |}
\newcommand\vka[1]{\left |{#1}\right\rangle}
\newcommand\ip[2]{\left\langle {#1},{#2} \right\rangle}

\newcommand\oP[2]{{#1}\oplus {#2} }
\newcommand\oM[2]{{#1}\ominus {#2} }
\newcommand\cA[1]{\mathcal {#1}}

\DeclareMathOperator{\ran}{ran}
\DeclareMathOperator{\Span}{span\,}
\DeclareMathOperator{\mo}{mod}
\DeclareMathOperator{\eff}{eff}
\DeclareMathOperator{\tr}{tr}

\makeatletter
\newcommand{\mathleft}{\@fleqntrue\@mathmargin0pt}
\newcommand{\mathcenter}{\@fleqnfalse}
\makeatother
\begin{document}
\begin{titlepage}
\title{Transition maps between Hilbert subspaces and quantum energy transport
}
\author{
\normalsize\textbf{Jorge R. Bola\~nos-Serv\'in, Roberto Quezada and Josu\'e I. Rios-Cangas}
\\
\small Departamento de Matem\'aticas \\[-1.6mm] 
\small Universidad Aut\'onoma Metropolitana, Iztapalapa Campus \\[-1.6mm] 
\small San Rafael Atlixco 186, 09340 Iztapalapa, Mexico City.\\[-1.6mm]
\small \texttt{jrbs@xanum.uam.mx}, \;
\small \texttt{roqb@xanum.uam.mx}, \; 
\small \texttt{jottsmok@xanum.uam.mx}
}
\date{\today}
\maketitle

\begin{center}
\begin{minipage}{5in}
\centerline{{\bf Abstract}} \medskip
We use  a natural generalization of the discrete Fourier transform to define transition maps between Hilbert subspaces and the global transport operator $Z$. By using  these transition maps as Kraus (or noise) operators, an extension of the quantum energy transport model of \cite{MR3860251} describing the dynamics of an open quantum system of $N$-levels is presented. We deduce the structure of the invariant states which can be recovered by transporting states supported on the first level.\end{minipage}
\end{center}
\thispagestyle{empty}
\end{titlepage}

\section{Introduction}\label{sec:intro}
In \cite{MR3399653} a two-level quantum model of the excitation energy transfer in quantum many-particle systems with a dipole interaction with a quantum field, was studied using the stochastic limit method to describe the dynamics, i.e., approximating the quantum dynamics with quantum random processes. Motivated by this work, a variation of this model that incorporates maximally entangled states and a transport operator, was studied in \cite{MR3860251,MR4107240} by framing it into the family of weak coupling limit type (WCLT) Gorini Kossakowski Su\-dhar\-san Linblad (GKSL) generators, with degenerate reference Hamiltonian.  In \cite{MR3860251} a characterization of all invariant states belonging to the commutant of the  system Hamiltonian was given, while in \cite{MR4107240} a structure of invariant states of the semigroup was shown to be a convex combination of two states, one supported on the interaction-free subspace $W$ and the other one on its orthogonal complement. In both cases a particular operator $Z$ was seen to govern the transport dynamics of the model and is fundamental in determining the structure of invariant states, reason why we regard it now as the transport operator.

The aim of this paper is to study an extension of the aforementioned models in where the system involves $N$-levels. It is worth pointing out that the transition maps used in those models are unitarily equivalent to a natural generalization of the discrete Fourier transform (DFT) between two Hilbert spaces, which is called transition map (see Definition \ref{def:transition-operator}). We emphasize that our approach generalizes the original model seen in \cite{MR3399653} as well as its variation in \cite{MR3860251}  (q.v. \cite{MR4107240}). Fixing an orthonormal basis of the Hilbert space, the transport operator $Z$ is defined in terms of the transition operators $Z_{k,k+1}$ between consecutive levels. By defining a suitable entangled basis $\{\varphi_{a_k} \,  : \, 0\leq a\leq n_{k}-1\}$ at each level $k$, a subspace $V$ is singled out to contain harmonic projections of the associated GKSL generator. We recall that the support projections of invariant states are subharmonic (see for instance \cite{MR1878987}). Thus, knowledge of subharmonic projections gives valuable insight on nontrivial invariant states. On the other hand, since the dynamics of  the invariant states supported on the interaction-free subspace $W$ is trivial, we restrict our attention to states supported on $\oM{V}{W}$ to give a full characterization of invariance. Namely, we show that a state $\rho$ supported on $\oM{V}{W}$  is invariant if and only if there exists a state $\tau$ supported on $P_1\oM{V}{W}$ such that $\rho$ may be recovered by moving $\tau$, via a CP map implemented by $Z$, throughout the whole subspace $V$ (see Theorems~\ref{th:characterisation-invariant-state} and~\ref{th:invariant-states-from-statesV1}). In sum, states supported on $P_1V$ completely determine the invariant states supported on $V$,  by means of  the transport operator $Z$.

The paper has a simple structure: Transition maps are defined in Section~\ref{sec:transition-maps} while the quantum transport model is described in Section~\ref{sec:N-energy-levels} To be more thorough, in Subsection~\ref{subs:two-levels} we define the transition operator between two mutually orthogonal subspaces (levels) of a Hilbert space by means of a rectangular DFT and derive some properties used in the sequel. We then restrict the transitions to consecutive levels and define the transport operator $Z$ in Subsection~\ref{subs:N-levels} In Subsection~\ref{sub-WCLT} we define the weak coupling limit type GKSL generator of the N-level extension of the model in \cite{MR3399653,MR3860251}. The main results of the paper are contained in Subsections~\ref{sub-harmonic} and~\ref{sub-invstates}, where some harmonic projections onto  suitable subspaces are studied and the structure of invariant states therein supported is deduced.

\section{Transition maps}
\label{sec:transition-maps}\setcounter{equation}{0}

\subsection{Two levels}
\label{subs:two-levels}

Let $\H$ be a finite dimensional complex Hilbert space and for $k, k'\in\N$, let $E_{k}, E_{k'}$ be mutually orthogonal subspaces of $\H$. Set $n_{k}\ceq\dim E_{k}$. We call $k, k'$ levels and denote by $\{\vk{a_k}\,:\,0\leq a\leq n_{k}-1\}$ the canonical orthonormal basis (onb) of the subspace $E_{k}$.

At each level, we consider the \emph{discrete Fourier transform},   
\begin{gather*}
F_k\ceq\frac{1}{\sqrt{n_k}}\sum_{a,a'=0}^{n_k-1}\zeta_k^{aa'}\vk{a'_k}\vb{a_k}\,,\quad k\geq 1
\end{gather*} 
where $\zeta_k\ceq e^{2\pi i/n_k}$ and $\vk{a'_k}\vb{a_k}$ denotes the rank-one operator from $E_{k}$ into itself defined as $\vk{a'_k}\vb{a_k}u = \langle a_{k}, u\rangle a'_{k}$. Coefficients $\zeta_k$ satisfy 
\begin{gather}
\label{eq:property-zetas}
\sum_{a=0}^{n_k-1}\zeta_k^{ja}=n_k\delta_{(j\mo n_k)\, 0}\,, \quad j=0,1,\dots
\end{gather} 
where $\delta_{ij}$ denotes the \emph{Dirac delta function}. 

Recall that $F_{k}$ is unitary and con-involutory on $E_{k}$: i.e., $F_{k}\cc{F}_{k}=I$, where 
\begin{gather*}
\cc{F}_k\ceq \frac{1}{\sqrt{n_k}}\sum_{a,a'=0}^{n_k-1}\cc{\zeta}_k^{aa'}\vk{a'_k}\vb{a_k}\,,\quad k\geq 1
\end{gather*}
 (cf.\cite[Probl.\, 2.2.P10]{MR2978290}).

Motivated by the above, the following is a natural generalization of the discrete Fourier transform.

\begin{definition}\label{def:transition-operator} The transition operator $Z_{k, k'}$ from $E_{k}$ into $E_{k'}$ is given  by 
\begin{gather}\label{eq:Gen-discreteFF}
Z_{k, k'}=\frac{1}{\sqrt{n_k}}\sum_{b=0}^{n_{k'}-1}\sum_{a=0}^{n_k-1}\zeta_k^{b a}\vk{b_{k'}}\vb{a_k}\,,\quad k, k' \geq 1
\end{gather} being $\{\vk{b_{k'}}\,:\,0\leq b\leq n_{k'}-1\}$ the canonical onb of $E_{k'}$.  
\end{definition}
Clearly, $Z_{k,k}=F_{k}$, $Z_{k, k'}^{2}=0$, $Z_{k,k'}^{*2}=0$, $P_{k'}Z_{k,k'}=Z_{k,k'}$ and $Z_{k,k'}P_{k}=Z_{k,k'}$,
where $P_{k}$ is  the orthogonal projection onto $E_{k}$. Besides, taking adjoints, $Z_{k,k'}^{*}P_{k'}=Z_{k,k'}^{*}$, $P_{k}Z_{k,k'}^{*}=Z_{k,k'}^{*}$. These relations will be used freely along this work. 
\begin{theorem}\label{properties-Zk-kprime}
For $n_{k'} \leq n_k$, the following statements hold true: 
\begin{enumerate}
\item\label{eq:p-zk-k01} $Z_{k,k'} Z_{k,k'}^{*}=P_{k'}$.
\item\label{eq:p-zk-k02} $Z_{k,k'}^{*}Z_{k,k'}$ is a projection and $Z_{k,k'}^{*}Z_{k,k'}\leq P_{k}$, with equality if and only if $k=k'$, viz. $ Z_{k,k'}^{*}Z_{k,k'}$ is a sub-projection of $P_{k}$.
\item\label{eq:p-zk-k03} $P_{k}-Z_{k,k'}^{*}Z_{k,k'}= \displaystyle \frac{1}{n_k}\sum_{a,a'=0}^{n_k-1}\sum_{b=n_{k+1}}^{n_k-1}\zeta_k^{b(a-a')}\vk{a'_k}\vb{a_k}\,.$
\item\label{eq:p-zk-k04} 
$Z_{k,k'}F_k^*=\displaystyle \sum_{a=0}^{n_{k'}-1}\vk{a_{k'}}\vb{a_k}\,.$
\end{enumerate}
\end{theorem}
\begin{proof}
In view of \eqref{eq:property-zetas}, 
\begin{align*}
Z_{k, k'}Z_{k,k'}^*&=\frac{1}{n_k}\sum_{b,b'=0}^{n_{k'}-1}\sum_{a,a'=0}^{n_k-1}\zeta_k^{ba-b'a'}\vk{b_{k'}}\delta_{a_ka'_k}\vb{b'_{k'}}\\&=\frac{1}{n_k}\sum_{b,b'=0}^{n_{k'}-1}\sum_{a=0}^{n_k-1}\zeta_k^{(b-b')a}\vk{b_{k'}}\vb{b'_{k'}}=\frac{1}{n_k}\sum_{b,b'=0}^{n_{k'}-1}n_k\delta_{bb'}\vk{b_{k'}}\vb{b'_{k'}}=P_{k'}\,.
\end{align*} This proves \ref{eq:p-zk-k01}.

Operator $Z_{k,k'}^*Z_{k,k'}$ is selfadjoint and by the above item,
\begin{gather*}
(Z_{k,k'}^*Z_{k,k'})^2=Z_{k,k'}^*Z_{k,k'}Z_{k,k'}^*Z_{k,k'}=Z_{k,k'}^*P_{k'}Z_{k,k'}=Z_{k,k'}^*Z_{k,k'}\,,
\end{gather*}
since $Z_{k,k'}^*P_{k'}= Z_{k,k'}^*$.
Thus, one has that $Z_{k,k'}^*Z_{k,k'}$ is a projection. Now, $P_k-Z_{k,k'}^*Z_{k,k'}$ is selfadjoint and 
\begin{align*}
(P_k-Z_{k,k'}^*Z_{k,k'})^2&=P_k^2+(Z_{k,k'}^*Z_{k,k'})^2-P_kZ_{k,k'}^*Z_{kk'}-Z_{k,k'}^*Z_{k,k'}P_k\\
&=P_k+Z_{k,k'}^*Z_{k,k'}-2Z_{k,k'}^*Z_{k,k'}=P_k-Z_{k,k'}^*Z_{k,k'}\,,
\end{align*}
which implies that $P_k-Z_{k,k'}^*Z_{k,k'}$ is a projection. Hence, $Z_{k,k'}^*Z_{k,k'} \leq P_k$ with equality if and only if $k'=k$. This asserts \ref{eq:p-zk-k02}. 

Now we have, 
\begin{align*}
Z_{k,k'}^*Z_{k,k'}&=\frac{1}{n_k}\sum_{a,a'=0}^{n_k-1}\sum_{b,b'=0}^{n_{k'}-1}\zeta_k^{ba-b'a'}\vk{a_{k}}\delta_{b_{k'}b'_{k'}}\vb{a'_{k}}\\
&=\frac{1}{n_k}\sum_{a,a'=0}^{n_k-1}\sum_{b=0}^{n_{k}-1}\zeta_k^{b(a-a')}\vk{a_{k}}\vb{a'_{k}}-\frac{1}{n_k}\sum_{a,a'=0}^{n_k-1}\sum_{b=n_{k'}}^{n_{k}-1}\zeta_k^{b(a-a')}\vk{a_{k}}\vb{a'_{k}}\,,
\end{align*}
whence by \eqref{eq:property-zetas}, item \ref{eq:p-zk-k03} holds true. 

Finally, one computes 
\begin{align*}
Z_{k,k'}F_k^*&=\frac{1}{n_k}\sum_{b=0}^{n_{k'}-1}\sum_{a,c,c'=0}^{n_k-1}\zeta_k^{ba-cc'}\vk{b_{k'}}\delta_{a_kc_k}\vb{c'_k}
\\&=\frac{1}{n_k}\sum_{b=0}^{n_{k'}-1}\sum_{c,c'=0}^{n_k-1}\zeta_k^{c(b-c')}\vk{b_{k'}}\vb{c'_k}=\frac{1}{n_k}\sum_{b=0}^{n_{k'}-1}\sum_{c'=0}^{n_k-1}n_k\delta_{bc'}\vk{b_{k'}}\vb{c'_k}\,,
\end{align*} which readily implies \ref{eq:p-zk-k04}.
\end{proof}
\begin{definition} The element $\varphi_{a_{k}}$ of the \emph{$k$-entangled} basis $\{\varphi_{a_k}\}_{a=0}^{n_k-1}$, is defined as the inverse discrete Fourier transform    
\begin{gather}\label{def-maximal-k-entangled}
\varphi_{a_k}\ceq F_k^*\vk{a_k}=\frac{1}{\sqrt{n_k}}\sum_{b=0}^{n_k-1}\zeta_k^{-ba}\vk {b_k}\,,
\end{gather}
 of the basic vector $\vk{a_{k}}\in E_{k},$ for $a=0,\dots,n_{k}-1$. 
\end{definition}

From the unitarity of $F_{k}$ the next result immediately follows.
\begin{corollary}\label{lem:max-entangled-basis}

The $k$-entangled basis is an orthonormal basis of $E_k$.
\end{corollary}
The $k$-canonical and $k$-entangled basis of $E_k$ are related by the discrete Fourier transform $\vk{a_k}=F_k\varphi_{a_k}$ and its inverse \eqref{def-maximal-k-entangled}.

By abuse of notation, we let $\abs Z_{k,k'}$ stand for $Z^*_{k,k'}Z_{k,k'}$. If $n_{k'} < n_{k}$, then $Z_{k,k'}$ has a nontrivial kernel. Indeed we have the following corollary. 
\begin{corollary}\label{cor:Ker-Zk-kprime} The transition operator satisfies 
\begin{gather}\label{eq:Ker-Zk-kprime}
\ker\abs Z_{k,k'}=\ker Z_{k,k'}=\Span\llb\varphi_{a_k}\,:\, a=n_{k'},\dots,n_k-1\rrb\,.
\end{gather}
\end{corollary}
\begin{proof}
From item \ref{eq:p-zk-k04} in Theorem~\ref{properties-Zk-kprime} one has
\begin{gather}\label{eq:isometry-Zk-kprime}
Z_{k,k'}\varphi_{a_k}=Z_{k,k'}F_k^*\vk{a_k}=\sum_{a'=0}^{n_{k'}-1}\vk{a'_{k'}}\delta_{a_{k}' a_k}=\begin{cases}
a_{k'}, & 0\leq a \leq n_{k'}-1\\
0, &  n_{k'}\leq a \leq n_{k}-1
\end{cases}
\end{gather}
which yields the second identity in \eqref{eq:Ker-Zk-kprime}. Clearly, $\ker Z_{k,k'}\subset \ker \abs Z_{k,k'}$ and if  $\abs Z_{k,k'}u=0$, then   
$0=Z_{k,k'}Z_{k,k'}^{*}Z_{k,k'}u=P_{k'}Z_{k,k'}u=Z_{k,k'}u$. Thus, $\ker \abs Z_{k,k'}\subset \ker Z_{k,k'}$, whence the proof is finished.  
\end{proof}

As a consequence of Corollary~\ref{cor:Ker-Zk-kprime} and since $\abs Z_{k,k'}$ is selfadjoint, we obtain the orthogonal decomposition  $E_{k}=\ran \abs Z_{k,k'}\oplus\ker\abs Z_{k,k'}$.
\begin{proposition}\label{Zk-isomorphism}
Operators ${Z_{k,k'}}\rEu{\ran \abs Z_{k,k'}}$ and $Z_{k,k'}^{*}$ are isometric isomorphisms between the subspaces $\ran \abs Z_{k,k'}\subset E_{k}$ and $E_{k'}$. 
\end{proposition}
\begin{proof}
From \eqref{eq:isometry-Zk-kprime} we get $Z_{k,k'}\varphi_{a_{k}}=|a_{k'}\rangle,$ for $a=0,\dots, n_{k'}-1$. Multiplication by $Z_{k,k'}^{*}$ on both sides yields $Z_{k,k'}^{*} \vk {a_{k'}} = \abs Z_{k,k'}\varphi_{a_{k}}= \varphi_{a_{k}}$, since $\abs Z_{k,k'}$ is a projection. Thus, orthonormal basis are sent into orthonormal basis. Hence, from the fact that $\ran \abs Z_{k,k'}=\Span \{\varphi_{a_{k}}\,:\, a=0, \dots, n_{k'}-1\}$, we conclude the assertion.
\end{proof}

The last proof shows that $Z_{k,k'}$ and $Z_{k,k'}^{*}$ perform transitions 
\begin{gather}\label{eq:basis-transitions-Zk}
Z_{k,k'}\varphi_{a_k}=\vk{a_{k'}}\qquad\mbox{and}\qquad 
Z_{k,k'}^*\vk{a_{k'}}=\varphi_{a_k}\,,
\end{gather}
between the entangled and canonical bases of  $\ran\abs Z_{k,k'}$ and $E_{k'}$, respectively. 
\begin{proposition}\label{prop:Z-ak}
The transition operator has the following properties, 
\begin{enumerate}
\item\label{itZ-ak01} $Z_{k,k'}\vk{a_{k}} = \displaystyle \frac{1}{\sqrt{n_k}}\sum_{a'=0}^{n_{k'} -1} \zeta_{k}^{a'a} |a_{k'}'\rangle$. Particularly, $
Z_{k,k'}|0_k\rangle = \Big(\frac{n_{k'}}{n_k}\Big)^{1/2} \varphi_{0_{k'}}$.
\item\label{itZ-ak02}$Z_{k', k''}Z_{k,k'}|a_{k}\rangle =\displaystyle \frac{1}{\sqrt{n_kn_{k'}}} \sum_{b=0}^{n_{k''} -1}\sum_{a'=0}^{n_{k'} -1} \zeta_{k}^{a'a} \zeta_{k'}^{ba'}|b_{k''}\rangle$. 
\end{enumerate}
\end{proposition}
\begin{proof}
It is immediately from \eqref{def-maximal-k-entangled} that 
$\vk{a_{k}} =n_k^{-1/2}\sum_{a'=0}^{n_{k} -1} \zeta_{k}^{a'a} \varphi_{a'_{k}}$. Hence, item \ref{itZ-ak01} follows after left multiplication by  $Z_{k,k'}$ and applying \eqref{eq:basis-transitions-Zk}, \eqref{eq:Ker-Zk-kprime}. 
Now, left multiplication of equation in item  \ref{itZ-ak01} by $Z_{k', k''}$ yields item \ref{itZ-ak02}. 
\end{proof}

The orthogonal projection of $\H$ onto $\ker\abs Z_{k,k'}$ in each level, is  
\begin{gather}
\label{eq:Orthogonal-Zk-kprime}
\abs Z_{k,k'}^{\perp}\ceq P_{k}-\abs Z_{k,k'}\,.
\end{gather}
\begin{proposition} \label{prop:explicit-rep-absZ} The following explicit representations hold true:
\begin{gather}\label{eq:explicit-rep-absZ}
\abs Z_{k,k'}= \sum_{b=0}^{n_{k'}-1}\vk{\varphi_{b_k}}\vb{\varphi_{b_k}}\qquad\mbox{and}\qquad  \abs Z_{k,k'}^\perp =\sum_{b=n_{k'}}^{n_{k}-1}\vk{\varphi_{b_k}}\vb{\varphi_{b_k}}\,.
\end{gather} 
\end{proposition}
\begin{proof}
From item \ref{eq:p-zk-k03} of Theorem~\ref{properties-Zk-kprime} and \eqref{def-maximal-k-entangled}, the projection $\abs Z_{k,k'}^{\perp}$ has the explicit form  
\begin{align*}
\abs Z_{k,k'}^{\perp}= & \frac{1}{n_k}\sum_{a,a'=0}^{n_k-1}\sum_{b=n_{k+1}}^{n_k-1}\zeta_k^{b(a-a')}\vk{a'_k}\vb{a_k} \\ = &
 \sum_{b=n_{k+1}}^{n_k-1}\vka{\frac{1}{\sqrt{n_k}}\sum_{a'=0}^{n_k-1}\zeta_k^{-ba'} a'_k}\vba{\frac{1}{\sqrt{n_k}}\sum_{a=0}^{n_k-1}\zeta_k^{-ba} a_k }= \sum_{b=n_{k+1}}^{n_{k}-1}\vk{\varphi_{b_k}}\vb{\varphi_{b_k}}\,.   
\end{align*}
This also implies the left-hand side of \eqref{eq:explicit-rep-absZ}.
\end{proof}
\subsection{$N$-levels}
\label{subs:N-levels}
We now consider $N$-levels, with $N\in\N$. Denote by $E_{k}$ the associated subspaces with orthonormal canonical basis $\{| a_{k}\rangle : a=0, \dots, n_{k}-1\}$, where $n_{k}=\dim E_{k}$, and orthogonal projections $P_{k}=\sum_{a=0}^{n_{k}-1}|a_{k}\rangle\langle a_{k}|$ of $\H$ onto $E_{k}$, for all $k=1,\dots,N$. So, $\H$ has the canonical basis 
\begin{gather*}
\{\vk{a_k}\,:\,0\leq a\leq n_k-1,\,1\leq k\leq N\}\,.
\end{gather*}
 
The only transitions  which will be considered are those between con\-se\-cu\-ti\-ve levels $k,k+1$, namely, transitions induced by the maps $Z_{k, k+1}$ for all $k=1, \dots, N-1$. For abbreviation, we write $Z_{k}$ instead of $Z_{k, k+1}$ and we assume that $n_{k+1}\leq n_{k}$, for all  $k=1, \dots, N-1$.

The \textit{transport } operator $Z$ is defined as  
\begin{gather}\label{eq:oplus-Z} 
Z\ceq\bigoplus_{k=1}^{N-1}Z_k\,.
\end{gather}
Hence, one has that  $ZP_k=Z_{k}$. The operator \eqref{eq:oplus-Z} maps  $\bigoplus_{k=1}^{N-1} P_{k}\H$ into $\H$ and the adjoint $Z^{*}= \bigoplus_{k=2}^{N}Z_{k}^{*}$ maps $\bigoplus_{k=2}^{N}P_{k}\H$ into $\H$.

Theorem~\ref{properties-Zk-kprime} readily implies 
\begin{gather}\label{eq:Z-Zstar}
ZZ^*=\bigoplus_{k=1}^{N-1}Z_kZ_k^*= \bigoplus_{k=2}^{N}P_{k}\quad \mbox{and}\quad Z^*Z=\bigoplus_{k=1}^{N-1}Z_k^*Z_k= \bigoplus_{k=1}^{N-1}\abs{Z_k}\,.
\end{gather}
In this fashion,  $|Z|=Z^*Z$ is a subprojection of $P=\bigoplus_{k=1}^{N-1} P_{k}$, i.e., $|Z|\leq P$. Besides, if $\abs Z_k\ceq \abs ZP_k$, then it is clear that $\abs Z_k=\abs {Z_k}$, for $k=1,\dots,N-1$.
 
If $n_{k+1}< n_{k}$ for some $k=1, \dots, N-1$, then  $Z$ has a nontrivial kernel. This is a consequence of the following proposition.
\begin{proposition}\label{cor:Ker-Z} 
The transport operator satisfies 
\begin{enumerate}
\item\label{it.core-Z01} 
$ \ker Z=\ker |Z|=\dS\bigoplus_{k=1}^{N-1}\Span\llb\varphi_{a_k}\,:\, a=n_{k+1},\dots,n_k-1\rrb$.
\item \label{it.core-Z02} $Z\rEu{\ran|Z|}$ and $Z^*\rEu{ZZ^*\H}$ are isometric isomorphisms from $\ran|Z|$ onto $ZZ^{*}\H$.
\item \label{it.core-Z03} For $a=0,\dots,n_{k+1}-1$ and $ k=1,\dots,N-1$,
\begin{gather}\label{eq:basis-transitions-Z}
Z\varphi_{a_k}=\vk{a_{k+1}}\qquad\mbox{and}\qquad 
Z^*\vk{a_{k+1}}=\varphi_{a_k}\,.
\end{gather}
\end{enumerate}
\end{proposition}
\begin{proof}
Item \ref{it.core-Z01} is a consequence of Corollary~\ref{cor:Ker-Zk-kprime}. Indeed, it is easily seen that $\ker Z\rEu {P_{k}\H}= \ker Z_{k}=\ker\abs Z_{k}$, for $ k=1, \dots, N-1$. Hence, item \ref{it.core-Z01} follows from  \eqref{eq:Ker-Zk-kprime}, afterwards replacing $k'=k+1$ and using  \eqref{eq:oplus-Z}. 

Now, one has by Proposition~\ref{Zk-isomorphism} that  ${Z_{k}}\rEu {\ran\abs Z_{k}}$ and $Z_{k}^{*}$ are isometric isomorphism between subspaces $\ran\abs Z_{k}$ and $E_{k+1}$, for $k=1, \dots, N-1$.  Hence, due to \eqref{eq:oplus-Z} and decomposition \eqref{eq:Z-Zstar}, item \ref{it.core-Z02} readily follows. 

Finally, item \ref{it.core-Z03} is a consequence of \eqref{eq:basis-transitions-Zk}, with $k'=k+1$.
\end{proof}

By virtue of  \eqref{eq:Z-Zstar} and \eqref{eq:Orthogonal-Zk-kprime}, we denote 
\begin{align}\label{eq:absZ-ortho}
\abs Z^{\perp}\ceq P-\abs Z = \bigoplus_{k=1}^{N-1} \abs Z_{k}^{\perp}\,.
\end{align}
\begin{corollary}
The following representations hold true, 
\begin{gather*}
\abs Z=\bigoplus_{k=1}^{N-1}\sum_{b=0}^{n_{k+1}-1}\vk{\varphi_{b_k}}\vb{\varphi_{b_k}}\qquad\mbox{and}\qquad |Z|^{\perp}=\bigoplus_{k=1}^{N-1}\sum_{b=n_{k+1}}^{n_{k}-1}\vk{\varphi_{b_k}}\vb{\varphi_{b_k}}\,.
\end{gather*} 
\end{corollary}
\begin{proof}
The proof is simple from Proposition~\ref{prop:explicit-rep-absZ}, \eqref{eq:Z-Zstar} and \eqref{eq:absZ-ortho}.
 \end{proof}

We shall partially address the matter of how the powers of $Z$ act on different levels.
\begin{corollary}\label{prop:powers-Z}
If $k\in\{1,\dots,N-1\}$, then for $m=1,2,\dots\leq (N-k)/2$, the following relations hold true:
\begin{align}\label{eq:0k-m-even}
Z^{2m-1}\vk{0_k}&=\prod_{j=0}^{m-1}\left(\frac{n_{k+2j+1}}{n_{k+2j}}\right)^{1/2}\varphi_{0_{k+2m-1}}
\,;\\[3mm]\label{eq:0k-m-odd}
Z^{2m}\vk{0_k}&=\prod_{j=0}^{m-1}\left(\frac{n_{k+2j+1}}{n_{k+2j}}\right)^{1/2}\vk{0_{k+2m}}\,.
\end{align}
\end{corollary}
\begin{proof}
For a fixed $k\in\{1,\dots,N-1\}$, the proof is carried out by induction on $m$ and we only need to show \eqref{eq:0k-m-even}, since \eqref{eq:0k-m-odd} follows from \eqref{eq:basis-transitions-Z}. From item \ref{itZ-ak01} of Proposition~\ref{prop:Z-ak}, with $k'=k+1$, one has  
\begin{align}\label{eq:01k-0}
\begin{split}
Z\vk{0_k}&=\left(\frac{n_{k+1}}{n_k}\right)^{1/2}\varphi_{0_{k+1}}\,.
\end{split}
\end{align}
We then may suppose that \eqref{eq:0k-m-even} holds true for $m-1$. Thus, a simple computation yields 
\begin{gather*}
Z^{2m-1}\vk{0_k}=ZZ^{2(m-1)}\vk{0_k}=\prod_{j=0}^{m-2}\left(\frac{n_{k+2j+1}}{n_{k+2j}}\right)^{1/2}Z\vk{0_{k+2(m-1)}}\,,
\end{gather*}
whence by virtue of \eqref{eq:01k-0}, we conclude the statement.
\end{proof}
\section{Application: energy transport in an open quantum system with $N$ energy levels}\setcounter{equation}{0}
\label{sec:N-energy-levels}
\subsection{N-level quantum transport model}
\label{sub-WCLT}
From now on, we enlarge the Hilbert space $\H$.  In addition to the conditions of  Subsection~\ref{subs:N-levels}, we introduce  two nondegenerate subspaces $E_{-}=\C\vk-$, $E_{+}=\C\vk+$, with associated energies $\varepsilon_{-}$, $\varepsilon_{+}$ and orthogonal projections $P_{-}=\vk-\vb-$, $P_{+}=\vk+\vb+$ of $\H$ onto $E_{\pm}$, respectively. All together, the Hilbert space now has the canonical orthonormal basis 
\begin{gather*}
\llb \vk-, \vk+, \vk{a_k}\,:\,0\leq a\leq n_k-1,\,1\leq k\leq N\rrb\,.
\end{gather*} 

It is well-known that the irreversible evolution $\rho \mapsto {\mathcal T}_{t}(\rho), \; t\geq 0$, of states of a quantum system interacting with environment (quantum open system) is described by a master equation 
\[\frac{d{\mathcal T}_{t}(\rho)}{dt}= {\mathcal L}\big({\mathcal T}_{t}(\rho)\big), \quad {\mathcal T}_{0}(\rho)=\rho\]
involving an infinitesimal generator ${\mathcal L}$ with the Gorini-Kossakowski-Sudarshan and Lindblad (GKSL) structure.  The family $({\mathcal T}_{t})_{t\geq 0}$ of completely positive operators acting on $L_{1}({\mathcal H})$ (the space of finite trace operators) is called \textit{quantum Markov semigroup} (QMS).

We shall consider a GKSL Markov generator ${\mathcal L}$ belonging to the class of weak coupling limit type with degenerate reference Hamiltonian
\begin{gather*}
H\ceq \varepsilon_-P_{-}+\varepsilon_+P_++\sum_{k=1}^N \varepsilon_kP_k \,.
\end{gather*}
The energies $\varepsilon_-,\varepsilon_+,\varepsilon_1,\dots,\varepsilon_N$ are supposed to be pairwise different and   
\begin{gather*}
\varepsilon_-<\varepsilon_{k+1}<\varepsilon_{k}<\varepsilon_+\,, \quad 1\leq k\leq N\,.
\end{gather*} 

We have $N+1$ positive \emph{Bohr frequencies} (q.v. \cite[Subsec.\,1.1.5]{MR1929788})
\begin{align*}
\omega_+&=\varepsilon_+-\varepsilon_1\\
\omega_k&=\varepsilon_k-\varepsilon_{k+1}\qquad	(1\leq k\leq N-1)\\
\omega_-&=\varepsilon_N-\varepsilon_-\,,
\end{align*}
which are assumed to be different. Levels and transitions can be arranged in a graph (see Fig. \ref{fig:graph}). 
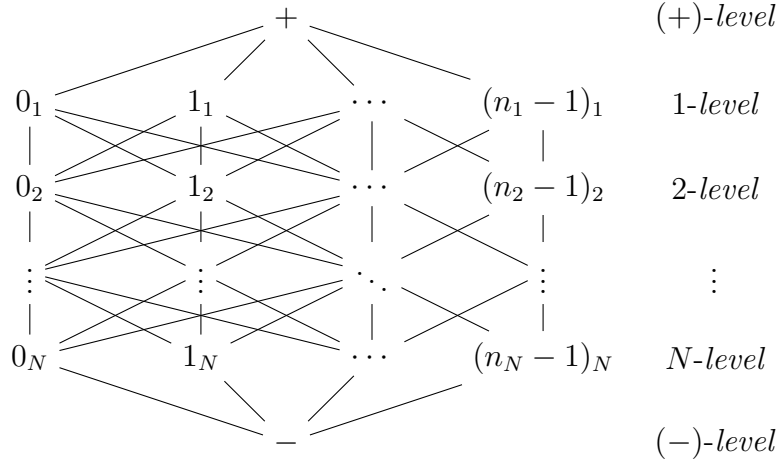
\begin{figure}[h]
\centering
\begin{tikzpicture}
  [scale=.75,auto=left,every node/.style={}]
  \node (n0) at (4.5,7.5) {+};
  \node (n1) at (0,6)  {$0_1$};
  \node (n2) at (3,6)  {$1_1$};
  \node (n3) at (6,6) {$\cdots$};
  \node (n4) at (9,6)  {$(n_1-1)_1$};
   \node (m1) at (0,4.5)  {$0_2$};
   \node (m2) at (3,4.5)  {$1_2$};
   \node (m3) at (6,4.5)  {$\cdots$};
   \node (m4) at (9,4.5)  {$(n_2-1)_2$};
   \node (p1) at (0,3)  {$\vdots$};
   \node (p2) at (3,3)  {$\vdots$};
   \node (p3) at (6,3)  {$\ddots$};
   \node (p4) at (9,3)  {$\vdots$};
    \node (N1) at (0,1.5)  {$0_N$};
   \node (N2) at (3,1.5)  {$1_N$};
   \node (N3) at (6,1.5)  {$\cdots$};
   \node (N4) at (9,1.5)  {$(n_N-1)_N$};
    \node (o0) at (4.5,0)  {$-$};
     \node (lM) at (12,7.5)  {$(+)$-\it level};
      \node (l1) at (12,6)  {$1$-\it level};
       \node (l2) at (12,4.5)  {$2$-\it level};
        \node (ld) at (12,3)  {$\vdots$};
         \node (lN) at (12,1.5)  {$N$-\it level};
          \node (lm) at (12,0)  {$(-)$-\it level};
  \foreach \from/\to in {n0/n1,n0/n2,n0/n3,n0/n4,o0/N1,o0/N2,o0/N3,o0/N4,n1/m1,n1/m2,n1/m3,n2/m1,n2/m2,n2/m3,n3/m1,n3/m2,n3/m3,n3/m4,n4/m3,n4/m4,
  m1/p1,m1/p2,m1/p3,m2/p1,m2/p2,m2/p3,m3/p1,m3/p2,m3/p3,m3/p4,m4/p3,m4/p4,
  p1/N1,p1/N2,p1/N3,p2/N1,p2/N2,p2/N3,p3/N1,p3/N2,p3/N3,p3/N4,p4/N3,p4/N4}
    \draw (\from) -- (\to);
\end{tikzpicture}
\caption{Graph of states and transitions}\label{fig:graph}
\end{figure}

The WCLT Markov generator ${\mathcal L}$ has the GKSL structure 
\begin{eqnarray*}
\begin{aligned}
\cA L(\rho)\ceq\sum_{\omega\in\atop{\{\omega_{-}, \omega_{+}, \omega_{k}:\atop\quad 1\leq k\leq N-1\}}} & -i[\Delta_\omega,\rho]+
  \left(L_{-, \omega} \rho L_{-, \omega}^{*}- \frac{1}{2} \llb L_{-, \omega}^* L_{-, \omega},\rho \rrb
  \right)  \\ &+
 \left(L_{+, \omega} \rho L_{+, \omega}^*-\frac{1}{2}\llb L_{+, \omega}^* L_{+, \omega}, \rho \rrb \right)
\end{aligned}
\end{eqnarray*}
and the dual generator 
\begin{eqnarray*}
\begin{aligned}
\cA L^*(x) \ceq \sum_{\omega\in\atop{\{\omega_{-}, \omega_{+}, \omega_{k}:\atop\quad 1\leq k\leq N-1\}}} & i[\Delta_\omega,x]+\left(L_{-, \omega}^* x L_{-, \omega}-\frac{1}{2}\llb L_{-, \omega}^*L_{-, \omega}x\rrb\right) \\ &+
\left(L_{+, \omega}^* x L_{+, \omega}-\frac{1}{2}\llb L_{+, \omega}^*L_{+, \omega} x \rrb\right)\,.
\end{aligned}
\end{eqnarray*}

The explicit form of the Kraus operators is given by  
\begin{align}\label{eq:Gamma-Krauss}
L_{-,\omega_+}&=\sqrt{n_1\Gamma_{-,\omega_+}}\vk{\varphi_{0_1}}\vb{+} &L_{+,\omega_+}&=\sqrt{n_1\Gamma_{+,\omega_+}}\vk{+}\vb{\varphi_{0_1}}\nonumber\\[3mm]
L_{-,\omega_k}&=\sqrt{\Gamma_{-,\omega_k}}Z_k &L_{+,\omega_k}&=\sqrt{\Gamma_{+,\omega_k}}Z_k^*\,, \quad 1\leq k\leq N-1\\[3mm]
L_{-,\omega_-}&=\sqrt{\Gamma_{-,\omega_-}}\vk{-}\vb{\varphi_{0_N}} &L_{+,\omega_-}&=0\nonumber
\end{align}  and the effective Hamiltonian is 
\begin{align*}
H_{\eff}\ceq \sum_{\omega} \Delta_{\omega}=\,&n_1\gamma_{-,\omega_+}P_+-n_1\gamma_{+,\omega_+}\vk{\varphi_{0_1}}\vb{\varphi_{0_1}}+\gamma_{-,\omega_-}\vk{\varphi_{0_N}}\vb{\varphi_{0_N}}\\&-\gamma_{+,\omega_-}P_-
+\sum_{k=1}^{N-1}\gamma_{-,\omega_k}\abs Z_k-
\gamma_{+,\omega_k}P_{k+1}\,.
\end{align*} 
The term $\gamma_{+,\omega_-}P_-$ is absent from the effective Hamiltonian in \cite{MR3860251}.

\subsection{Harmonic projections and invariant states}
\label{sub-harmonic}
The \emph{quantum Markov semigroup} $({\cA T}^*_{t})_{t\geq 0}$ of operators acting on the von Neumann algebra $B(\H)$ of all bounded operators on $\H$, is the adjoint semigroup of $({\mathcal T}_{t})_{t\geq 0}$ defined by the duality relation between states $\rho$ and observables $x$, 
\begin{gather*}
\tr \left(\rho{\cA T}_{t}^{*}(x)\right)=\tr\left({\cA T}_{t}(\rho)x\right)\,.
\end{gather*}

A selfadjoint operator $p\in\cA B(\H)$ is said to be \emph{subharmonic}  for $({\cA T}^*_{t})_{t\geq 0}$ if 
\begin{gather}\label{eq:subharmonic}
\cA T_{t}^{*}(p)\geq p\,,\quad\mbox{ for all }t\geq0\,.
\end{gather} 
We say that $p$ is \emph{superharmonic} if the inequality \eqref{eq:subharmonic} is reversed and \emph{harmonic} if the equality takes place, i.e., when $\cA L^*(p)=0$.

Subharmonic projections are deeply related with the stationary states of a  quantum Markov semigroup. Indeed, any support projection of an invariant state $\rho$, viz. the orthogonal projection of $\H$ onto the closure of $\ran\rho$, is subharmonic. Although the converse is not true, knowledge of the subharmonic projections gives useful insight on the invariant states.   
\begin{proposition}\label{lem:harmonic-perp}
A projection $P_E$ onto a subspace $E\subset \H$ is subharmonic if and only if the projection $P_{E^\perp}$ onto $E^\perp$ is superharmonic. Moreover, if $P_E$ is harmonic then so is $P_{E^\perp}$.
\end{proposition}
\begin{proof}
If $P_E$ is subharmonic then for $t\geq0$, using the identity preserving property, $\cA T_{t}^{*}(I)=I$, of the dual semigroup, we get  
\begin{align}\label{eq:superharmonic-property}
\cA T_{t}^{*}(P_{E^\perp})=\cA T_{t}^{*}(I)-\cA T_{t}^{*}(P_{E})\leq I-P_{E}=P_{E^\perp}\,,
\end{align}
which implies $P_{E^\perp}$ is superharmonic. We  prove the converse assertion, interchanging the roles of $P_{E},P_{E^\perp}$ and the inequality in \eqref{eq:superharmonic-property}. Now, if $P_{E}$ is harmonic then the equality \eqref{eq:superharmonic-property} holds true and, hence, $P_{E^\perp}$ is harmonic. 
\end{proof}

Conditions characterizing sub-harmonic projections were given in Theorem III.1 of \cite{MR1878987}. The proof that those conditions are sufficient for an operator to be harmonic is rather simple.    
\begin{lemma}\label{prop:harmonic-property-D}
A selfadjoint operator $p\in\cA B(\H)$ is harmonic for $({\mathcal T}_{t}^{*})_{t\geq 0}$ if    
\begin{gather}\label{eq:DKraus-subharmonic}
L_{-,\omega_{k}}p=pL_{-, \omega_{k}}p\,;\quad L_{-,\omega_{k}}^*p=pL_{-, \omega_{k}}^*p\,;\quad k=-, +, 1, \dots, N-1
\end{gather} 
\end{lemma}
\begin{proof}
Relations \eqref{eq:DKraus-subharmonic} readily imply that $p$ commutes with all Kraus o\-pe\-ra\-tors $L_{\pm, \omega_{k}}$ and $H_{\eff}$. Therefore, one checks at once that $\cA L^*(p)=0$. This implies that $p$ is a fixed point for $({\mathcal T}_{t})_{t\geq 0}$, hence harmonic. 
\end{proof}

Let us set  
\begin{align}\label{eq:calV-harmonic}
\begin{split}
 V\ceq\{ &\vk-, \vk+, Z^n\varphi_{0_1}, {Z^*}^n\varphi_{0_N},{Z^*}^{s_m}\varphi_{0_{2m+1}}\,: \\ & 0\leq n\leq N-1, \, 1\leq m  \leq (N-1)/2,\, 1\leq s_m\leq 2m{\}}^\perp\,.
\end{split}
\end{align}

\begin{theorem} The projection $P_{V^\perp}$ of $\H$ onto $V^\perp$ is harmonic and, hence, so is the projection $P_V$ onto $V$.
\end{theorem}
\begin{proof} 
From Lemma~\ref{prop:harmonic-property-D}, we only need to show \eqref{eq:DKraus-subharmonic}, which in view of \eqref{eq:Gamma-Krauss}, holds if and only if, for all $u\in V^\perp$ and $k=1,\dots,N-1$,
\begin{enumerate}
\begin{multicols}{2}
\item\label{eq:01K-r} $\ip{+}{u}\varphi_{0_1}=\ip{+}{u}P_{ V^\perp}\varphi_{0_1}$
\item\label{eq:02K-r} $\ip{\varphi_{0_1}}{u}\vk+=\ip{\varphi_{0_1}}{u}P_{ V^\perp}\vk+$
\item \label{eq:03K-r}$\ip{\varphi_{0_N}}{u}\vk-=\ip{\varphi_{0_N}}{u}P_{ V^\perp}\vk-$
\item\label{eq:04K-r} $\ip{-}{u}\varphi_{0_N}=\ip{-}{u}P_{ V^\perp}\varphi_{0_N}$
\item\label{eq:05K-r} $Z_ku=P_{ V^\perp}Z_ku$
\item\label{eq:06K-r} $Z_k^*u=P_{ V^\perp}Z_k^*u$
\end{multicols}
\end{enumerate}
We see at once that the conditions \ref{eq:01K-r}-\ref{eq:04K-r} hold true. It remains to prove \ref{eq:05K-r}-\ref{eq:06K-r}, which will be satisfied once we show that $Z_ku,Z_k^*u\in V^\perp$, for all $u\in V^\perp$.  

It is immediate that $Z_ku\in V^\perp$, for all $u\in\Span\{\vk-,\vk+,Z^n\varphi_{0_1}\}$. From \eqref{eq:basis-transitions-Z} and Corollary~\ref{prop:powers-Z} it follows that  with $\beta(m,k)=\prod_{j=0}^{m-1}\left(\frac{n_{k+2j+1}}{n_{k+2j}}\right)^{-1/2}$ 
\begin{gather}\label{eq:power-Z-varphi-odd}
\varphi_{0_{2m+1}}=Z^{*}|0_{2m+2}\rangle = Z^{*} \beta(m,2) Z^{2m}|0_{2}\rangle = \beta(m,2)Z^{2m}\varphi_{0_1}\in V^\perp 
\end{gather}
Thus, for $k=2m+1-s_m$ (otherwise is zero) 
\begin{gather*}
 Z_k{Z^*}^{s_m}\varphi_{0_{2m+1}}=P_{k+1}{Z^*}^{(s_m-1)}\varphi_{0_{2m+1}}={Z^*}^{(s_m-1)}\varphi_{0_{2m+1}}\in  V^\perp\,.
\end{gather*}
Analogously, it can be shown that $Z_k{Z^*}^n\varphi_{0_N}\in V^\perp$. Therefore, we have that $Z_k^*u\in V^\perp$, for all $u\in\Span \{\vk-,\vk+,{Z^*}^n\varphi_{0_N},{Z^*}^{s_m}\varphi_{0_{2m+1}}\}$. Besides, for $n\neq0$ and $k=n+1$  (otherwise is trivial), if $n$ is even, we have again from \eqref{eq:basis-transitions-Z} and Corollary~\ref{prop:powers-Z}
\begin{gather*}
Z_k^*Z^{n}\varphi_{0_1}=Z_k^*Z^{n-1}\vk{0_2}=\beta(n/2,2)^{-1} Z_k^*\varphi_{0_{n+1}}\in V^\perp
\end{gather*}
 For the case $n$ odd, we use  again Corollary~\ref{prop:powers-Z} to show that for $k=n$ (otherwise is zero)
\begin{gather*}
Z_k^*Z^{n-1}\vk{0_2}= \beta\big((n-1)/2, 2\big)^{-1} Z_k^*\vk{0_{n+1}}=\beta\big((n-1)/2, 2\big)^{-1}\varphi_{0_{n}} \in V^\perp
\end{gather*}
this finishes the proof. Proposition~\ref{lem:harmonic-perp} implies that $P_V$ is harmonic.
\end{proof}

 The \emph{interaction-free} subspace defined by 
 \begin{gather*}
 W\ceq \bigcap_{{\omega\in\{\omega_{-}, \omega_{+}, \omega_{k}:\atop\:\:\: 1\leq k\leq N-1\}}}\left(\ker L_{\pm, \omega} \cap \ker L^*_{\pm, \omega}\right) \,,
 \end{gather*}
plays an important role in the structure of invariant states. An operator $x$ is supported on a subspace $E$ if $\cc{\ran x}\subset E$. 
\begin{remark}\label{rm:harmonic-invariant-W}
 It is a simple matter to see that a selfadjoint operator $x\in\cA B(\H)$ supported on $W$ satisfies $x= P_W x=x P_W$, where $P_W$ is the projection of $\H$ onto $W$. One easily verify that such an $x$ is harmonic and, moreover, any state $\rho$ supported on $W$ is $invariant$. 
 \end{remark}
 \begin{proposition}
The following identity holds true, 
\begin{eqnarray}\label{eq:characterization-Wd}
W=\ran \abs Z_1^\perp\,.
\end{eqnarray}
\end{proposition}
\begin{proof}
In view of \eqref{eq:Gamma-Krauss}, one verifies that $\ker L_{+,\omega_-}=\ker L^*_{+,\omega_-}=\H$ and 
\begin{align}\label{eq:Ker-DmDM}
\begin{split}\ker L_{-,\omega_-}\cap\ker L^*_{-,\omega_-}&=\llb\vk -,\varphi_{0_N}\rrb^\perp\,;\\
\ker L_{\pm,\omega_+}\cap\ker L^*_{\pm,\omega_+}&=\llb\vk+,\varphi_{0_1}\rrb^\perp\,.
\end{split}
\end{align}
Moreover, for $k=1,\dots,N-1$, one has $\ker L_{+, \omega_k}=\ker L^*_{-, \omega_k}=(P_{k+1}\H)^\perp$ and taking into account \eqref{eq:Ker-Zk-kprime}, we get  
\begin{gather*}
\ker L_{-, \omega_k}=\ker L^*_{+, \omega_k}=\oP{\Span\llb\vk-,\vk+\rrb}{\Span\llb\varphi_{t_k}\rrb_{t=n_{k+1}}^{n_k-1}}\hspace{-.5em}\bigoplus_{\tiny
\begin{array}{c}
   j=1\\j\neq k
  \end{array}}^N\hspace{-.3em}P_k\H\,.
\end{gather*} 
Thereby,
\begin{gather*}
\ker L_{\pm,\omega_k}\cap\ker L^*_{\pm,\omega_k}=\oP{\Span\llb\vk-,\vk+\rrb}{\Span\llb\varphi_{t_k}\rrb_{t=n_{k+1}}^{n_k-1}}\hspace{-.5em}\bigoplus_{\tiny
\begin{array}{c}
   j=1\\j\neq k,k+1
  \end{array}}^N\hspace{-.5em}P_k\H\,.
\end{gather*}
Hence, with \eqref{eq:Ker-DmDM},  
$W=\Span\{\varphi_{t_1}\}_{t=n_2}^{n_1-1}$ and Corollary~\ref{cor:Ker-Zk-kprime} implies \eqref{eq:characterization-Wd}.
\end{proof}
 The last result claims \begin{gather*}
 W^\perp=\oP{\Span\llb\vk-,\vk+\rrb}{\ran \abs Z_1}\bigoplus_{k=2}^N\ran P_{k}\,.
 \end{gather*}

The following assertion is adapted from \cite[Th.\,3.2]{MR4107240}.
\begin{proposition}\label{prop:convex-decomposition} 
A state $\rho$ commuting with $\abs Z_{1}^{\perp}$ is invariant if and only if there exist unique $\lambda\in[0,1]$ and invariant states $\eta,\tau\in\cA B(\H)$ which  commute with $\abs Z_1^\perp$ and supported on $W, W^\perp$, respectively, such that $\rho=\lambda \eta+(1-\lambda)\tau$.
\end{proposition}

In what follows, we shall work on the subspace $\Omega\ceq\{\vk-,\vk+, \varphi_{0_1},\varphi_{0_N}\}^\perp$ which contains $V$. One easily verifies  that if $\rho$ is a state supported on $\Omega$, then
\begin{align}\label{eq:predual-generator-invariant}
\begin{split}
\cA L(\rho)=\sum_{j=1}^{N-1}&\left(\eta_{-,\omega_j}\abs{Z}_j-
\eta_{+,\omega_j}P_{j+1}\right)\rho+\rho\left(\cc{\eta}_{-,\omega_j}\abs{Z}_j-
\cc{\eta}_{+,\omega_j}P_{j+1}\right)\\
&+\Gamma_{-,\omega_j}Z_j\rho Z_j^*+\Gamma_{+,\omega_j}Z_j^*\rho Z_j\,,
\end{split}
\end{align} with $\eta_{-,\omega_j}=-(\Gamma_{-,\omega_j}+2i\gamma_{-,\omega_j})/2$ and $\eta_{+,\omega_j}=(\Gamma_{+,\omega_j}-2i\gamma_{+,\omega_j})/2$. Besides,
\begin{gather}\label{eq:ker-rho-inOmega}
\llb\vk-,\vk+, \varphi_{0_1},\varphi_{0_N}\rrb\subset\ker \rho \quad \textrm{and} \quad \rho\sum_{k=1}^NP_k= \rho I=I  \rho=\sum_{k=1}^NP_k\rho 
\end{gather} where $I$ is the identity matrix.
From now on, we denote $\Gamma_{-,\omega_k}/\Gamma_{+,\omega_k}=e^{\beta_k}$,  with $\beta_k\ceq{\omega_k\beta(\omega_k)}$, for $k=1,\dots,N-1$.
\begin{theorem}\label{th:invariant-states}
Let  $\rho$ be a state supported on $\Omega$. Then $\rho$ is invariant if and only if  it commutes with $\abs Z_k, P_{k+1}$ and satisfies 
\begin{gather}\label{eq:detail-balance-state}
Z_k^*\rho Z_k=e^{\beta_k} \rho \abs Z_k\,,
\end{gather}
for all $k=1,\dots,N-1$. In this case $\rho$ belongs to $\llb H\rrb'$.
\end{theorem}
\begin{proof}
If $\rho$ is invariant then, by \eqref{eq:predual-generator-invariant} and since $P_k=\oP{\abs{Z}_k}{\abs Z^\perp_k}$,
\begin{equation*}
\begin{aligned}
0=P_k\cA L(\rho)P_N={}&\big((\eta_{-,\omega_k}-\eta_{+,\omega_{k-1}}-\cc{\eta}_{+,\omega_{N-1}})\abs{Z}_k\\&-(\eta_{+,\omega_{k-1}}+\cc{\eta}_{+,\omega_{N-1}})\abs Z^\perp_{k}\big)\rho P_N\,,
\end{aligned}\quad (k=1,\dots,N-1) 
\end{equation*}
with nonzero coefficients of $\abs{Z}_k$ and $\abs Z^\perp_k$. Due to orthogonality,  one has $\abs{Z}_k\rho P_N=\abs Z^\perp_k\rho P_N=0$, i.e.,  $P_k\rho P_N=P_N\rho P_k=0$. Thus, by \eqref{eq:ker-rho-inOmega}
\begin{gather}\label{eq:commutation-pn}
[\rho,P_N]=\sum_{k=1}^{N-1} P_k\rho P_N-P_N\rho\sum_{j=k}^{N-1} P_k=0\,,
\end{gather}
thence $\rho$ commutes with $P_N$. Besides, for $k,l=1,\dots,N-1$, with $k\neq l$,     
\begin{align*}
0=P_k\cA L(\rho)\abs{Z}_{l}={}&\big((\eta_{-,\omega_k}-\eta_{+,\omega_{k-1}}+\cc{\eta}_{-,\omega_{l}}-\cc{\eta}_{+,\omega_{l-1}})\abs{Z}_k
\\&-(\eta_{+,\omega_{k-1}}-\cc{\eta}_{-,\omega_{l}}+\cc{\eta}_{+,\omega_{l-1}})\abs Z^\perp_{k}\big)\rho\abs{Z}_{l}\,,
\end{align*}
whence $\abs{Z}_k\rho\abs{Z}_{l}=\abs Z^\perp_k\rho\abs{Z}_{l}=0$.  Similarly, identity $0=P_k\cA L(\rho)\abs Z^\perp_{l}$ implies that $\abs{Z}_k\rho \abs Z^\perp_l=\abs Z^\perp_k\rho \abs Z^\perp_l=0$. Hence, $P_k\rho P_l=0$. Thus, one obtains analogously to \eqref{eq:commutation-pn} that  $\rho$ commutes with $P_k$. So, $\rho\in\{H\}'$, since the commutation of $\rho$ with $P_+,P_-$ are evident.  Now, \eqref{eq:predual-generator-invariant} implies 
\begin{eqnarray}\label{eq:absZ-L-absZ}
\begin{aligned}
0=\abs Z_k \cA L(\rho)\abs Z_k ={}&\left(-\Gamma_{-,\omega_k}-\Gamma_{+,\omega_{k-1}}\right)\abs Z_k \rho\abs Z_k\\&+\Gamma_{-,\omega_{k-1}}\abs Z_k Z_{k-1}\rho Z_{k-1}^*\abs Z_k +\Gamma_{+,\omega_{k}}Z_k^*\rho Z_k\,.
\end{aligned}
\end{eqnarray}
So, we claim that 
\begin{gather}\label{eq:aux-induction-P_k}
Z_k^*\rho Z_k=e^{\beta_k} \abs Z_k\rho \abs Z_k\,,\quad k=1,2,\dots, N-1\,.
\end{gather}
Indeed, $Z_1^*\rho Z_1=e^{\beta_1}\abs Z_1\rho \abs Z_1$ and, by induction on $k$, assuming \eqref{eq:aux-induction-P_k} true for $k-1$, then left multiplication by $Z_{k-1}$, right multiplication by $Z_{k-1}^{*}$ and the fact that $\rho$ commutes with $P_k$ yield, 
\begin{gather}\label{eq:aux-induction-P_k-rho}
Z_{k-1}\rho Z_{k-1}^*=e^{-\beta_{k-1}}Z_{k-1}(Z_{k-1}^*\rho Z_{k-1})Z_{k-1}^*=e^{-\beta_{k-1}}\rho P_k\,.\end{gather}
Thus, replacing \eqref{eq:aux-induction-P_k-rho} in \eqref{eq:absZ-L-absZ},  one obtains \eqref{eq:aux-induction-P_k}.
Thereby, \eqref{eq:detail-balance-state} follows by \eqref{eq:aux-induction-P_k}, after proving that $\rho$ commutes with $\abs Z_k$, for all $k=1,\dots,N-1$. So, \begin{align*}
0=\abs Z_k \cA L(\rho)\abs Z^\perp_k=&{}\left(\eta_{-,\omega_k}-\Gamma_{+,\omega_{k-1}}\right)\abs Z_k \rho \abs Z^\perp_k\\&+\Gamma_{-,\omega_{k-1}}\abs Z_k Z_{k-1}\rho Z_{k-1}^*\abs Z^\perp_k\,,
\end{align*}
which by replacing \eqref{eq:aux-induction-P_k-rho}, one obtains $\abs Z_k \rho \abs Z^\perp_k= \abs Z^\perp_k\rho\abs Z_k=0$. Hence, since $\rho$ and $P_k$ commute,  
\begin{align*}
[\rho,\abs Z_k]&=\rho P_k\abs Z_k-\abs Z_k P_k\rho\\&=(\abs Z_k+\abs Z^\perp_k)\rho\abs Z_k-\abs Z_k\rho(\abs Z_k+\abs Z^\perp_k)=0\,.
\end{align*}

Conversely, left multiplication by $Z_{k-1}$, right multiplication by $Z_{k-1}^{*}$ and the fact that $\rho$ commutes with $P_{k+1}$ show that \eqref{eq:detail-balance-state} implies, 
\begin{gather}\label{eq:equivalency-rho-absZ}
Z_k\rho Z_k^*=e^{-\beta_k}\rho P_{k+1}\, \qquad(k=1,\dots,N-1)
\end{gather}
Thus, by replacing both equations in  \eqref{eq:predual-generator-invariant}, one gets that $\rho$ is invariant.
\end{proof}
\begin{remark}
For invariant states  supported on $\Omega$, properties \eqref{eq:detail-balance-state} and   
\eqref{eq:equivalency-rho-absZ} are equivalent.
\end{remark}
\begin{corollary}\label{cor:detailed-balance-property}
If $\rho$ is an invariant state supported on $\Omega$, then 
\begin{align}\label{eq:invariant-property}
\rho Z_k=e^{\beta_k} Z_k\rho\,,\quad \mbox{for }k=1,\dots,N-1\,.
\end{align}
\end{corollary}
\begin{proof}
From Theorem~\ref{th:invariant-states}, $\rho$ commutes with $\abs Z_k$ and $P_{k+1}$. Hence, \eqref{eq:detail-balance-state} implies 
$\rho Z_k=\rho P_{k+1}Z_k=Z_kZ_k^*\rho Z_k=e^{\beta_k} Z_k\abs Z_k\rho$,
which yields  \eqref{eq:invariant-property}.
\end{proof}
\begin{remark} As a consequence of Corollary~\ref{cor:detailed-balance-property} and by \eqref{eq:ker-rho-inOmega}, an invariant state $\rho$ supported on $\Omega$ satisfies $\rho L_{-,\omega_k}=e^{\beta_k} L_{-,\omega_k}\rho$, for $k=+,-,1,\dots,N-1$, i.e., $\rho$ is \emph{detailed balance} (cf. \cite[Sec.\,3.2]{MR4107240}).
\end{remark}
\subsection{Structure of invariant states}
\label{sub-invstates}
We shall turn out our attention to the subspace $ V\subset\Omega$,  given in \eqref{eq:calV-harmonic}.
\begin{proposition}\label{prop:omega-vs-V}If an invariant state $\rho$ is supported on $\Omega$ then so is on $ V$. 
\end{proposition}
\begin{proof}
For $n=0,\dots,N-1$, by virtue of \eqref{eq:invariant-property}, there exists $\alpha_n>0$ such that $Z^{*n}\rho=\alpha_n\rho Z^{*n}$. In this fashion, in view of \eqref{eq:ker-rho-inOmega}, for any $f\in\H$,
\begin{gather}\label{eq:invarian-in-V}
\ip{\rho f}{Z^n\varphi_{0_1}}=\alpha_n\ip{\rho Z^{*n} f}{\varphi_{0_1}}=0\,.
\end{gather}
Similarly, $\ip{\rho f}{Z^{*n}\varphi_{0_N}}=0$. Moreover, with suitable $s_m,m\in\N$, \eqref{eq:power-Z-varphi-odd} and \eqref{eq:invarian-in-V} imply  $\ip{\rho f}{{Z^*}^{s_m}\varphi_{0_{2m+1}}}=\alpha_{s_m}\ip{\rho Z^{s_m}f}{Z^{2m}\varphi_{0_1}}=0$, with $\alpha_{s_m}>0$. Hence, $V^\perp\subset\ker \rho$,  which implies our assertion. 
\end{proof}
\begin{remark}\label{rm:W-in-V}
The interaction-free subspace $W$ is contained in $ V$. Indeed, if $P_W$ is the projection of $\H$ onto $W\subset\Omega$, then $(\tr P_W)^{-1} P_W$ is an invariant state supported on $\Omega$, hence, by  Proposition~\ref{prop:omega-vs-V} one infers that $W\subset  V$.
\end{remark}
Let us make the following stratification of $V$. For $k=1,\dots,N$, we consider the subspaces $V_k\ceq P_kV$. Explicitly,   
\begin{gather*}
V_k=P_{k}\H \ominus\llb Z^{k-1}\varphi_{0_1}, {Z^*}^{N-k}\varphi_{0_N},{Z^*}^{2m-k+1}\varphi_{0_{2m+1}}\,:\,1\leq m \leq \frac{(k-1)}{2}\rrb\,.
\end{gather*}
It is evident from Remark~\ref{rm:W-in-V} that $W\subset V_1$.
\begin{lemma}\label{lem:properties-trasport-operator}
For $k=1,\dots,N-1$, the following statements are true.
\begin{enumerate}
\item\label{it:isom-V01} The transport operator satisfies 
\begin{gather}\label{eq:ZV-properties01}
Z^jV_k=V_{k+j}\,,\quad j=1,\dots,N-k-1\,.\\\label{eq:ZV-properties02}
\mbox{Moreover,}\hspace{3.5em} Z^kV=\bigoplus_{j=k+1}^{N}V_{j}\qquad\mbox{and}\qquad V=\bigoplus_{j=0}^{N-1}Z^jV_1\,.{\hspace{2.5em}}\end{gather}
\item\label{it:isom-V02} $Z_k\colon\abs Z_k V\to V_{k+1}$ and $Z^*_k\colon V_{k+1}  \to \abs Z_k V$ are isometric isomorphisms.
\end{enumerate}
\end{lemma}
\begin{proof} We check at once that $ZV_k^\perp=V_{k+1}^\perp$, where $V_k^\perp=\oM{P_k\H}{V_k}$, which implies $ZV_k=V_{k+1}$. Thus, by induction on $j$, suppose $Z^{j-1}V_k=V_{k+j-1}$ and thence 
$Z^jV_k=ZZ^{j-1}V_k=ZV_{k+j-1}$, which yields \eqref{eq:ZV-properties01}. Besides, 
\begin{align*}
Z^kV=\bigoplus_{j=1}^{N}Z^kV_{j}=\bigoplus_{j=1}^{N-k}V_{j+k}=\bigoplus_{j=k+1}^{N}V_{j}\,.
\end{align*}
Moreover, $\bigoplus_{j=0}^{N-1}Z^jV_1=\bigoplus_{j=0}^{N-1}V_{j+1}=V$, which proves the statement \ref{it:isom-V01}.

Now, by virtue of \eqref{eq:ZV-properties01} and item \ref{it.core-Z01} of Proposition~\ref{cor:Ker-Z},  
\begin{align*}
V_{k+1}=ZP_k V=Z_k(\abs Z_k+\abs Z_k^\perp)V=Z_k\abs Z_kV\,.
\end{align*}
This expression satisfies $Z_k^*V_{k+1}=\abs Z_kV$, whence from item \ref{it.core-Z02} of Proposition~\ref{cor:Ker-Z}, the statement \ref{it:isom-V02} readily follows.
\end{proof}

On account of the second identity in \eqref{eq:ZV-properties02},
we are motivated to give a representation for the subspace $V_1$.
\begin{proposition}
The subspace $V_1$ can be written as $V_1=\oP{W}{M}$, where 
\begin{align*}
M=\Span\Big\{&\frac{1}{\cc{\mu}_{b_1}}\varphi_{(b_1)_1}-\frac{1}{\cc{\mu}_{b_1+1}}\varphi_{(b_1+1)_1}\,:\\
&\mu_{b_1}=\prod_{j=2}^Nn_j^{-\frac12}\sum_{b_{N-1},\dots,b_{2}=0}^{n_{N}-1,\dots,n_3-1}\zeta_{N-1}^{-b_{N-2}b_{N-1}}\dots\zeta_{2}^{-b_{1}b_{2}}\in\C \Big\}_{b_1=1}^{n_2-2}
\end{align*}
\end{proposition}
\begin{proof}
We see at once that $V_1=\oM{P_{1}\H} {\{\varphi_{0_1},{Z^*}^{N-1}\varphi_{0_N}\}}$ and inasmuch as $W\subset V_1$, one has $\dim \big(\oM{V_1}{W}\big)=\dim M$. Thus, we only need to show that $M\subset V_1$. It follows from \eqref{def-maximal-k-entangled}, \eqref{eq:basis-transitions-Z} that 
$Z^*\varphi_{0_N}=n_N^{-1/2}\sum_{b=0}^{n_N-1}\varphi_{b_{N-1}}$, and iterating we get
$Z^{*2}\varphi_{0_N}=(n_N n_{N-1})^{-1/2}\sum_{b_2, b_1 =0}^{n_{N}-1, n_{N-1}-1}\zeta_{N-1}^{-b_1 b_2}\varphi_{(b_1)_{N-2}}$ Then, by iterating $(N-1)$-times,   
\begin{align*}
Z^{*N-1}\varphi_{0_N}=\! \prod_{j=2}^Nn_j^{-\frac12}\!\sum_{b_{N-1},\dots,b_{1}=0}^{n_{N}-1,\dots,n_2-1}\!\!\zeta_{N-1}^{-b_{N-2}b_{N-1}}\dots\zeta_{2}^{-b_{1}b_{2}}\varphi_{(b_1)_1}\!=\!\!\sum_{b_1=0}^{n_2-1}\mu_{b_1}\varphi_{(b_1)_1}\,.
\end{align*}
Now, $
\ip{\varphi_{a_1}/\cc{\mu}_{a}-\varphi_{(a+1)_1}/\cc{\mu}_{a+1}}{\sum_{b_1=0}^{n_2-1}\mu_{b_1}\varphi_{(b_1)_1}}
=0$, for $a=1,\dots,n_2-2$,
i.e., $M\perp \{ Z^{*N-1}\}$. Of course, $M\perp\{\varphi_{0_1}\}$, which finishes the proof. 
\end{proof}

 For any subspace $A_1\subset\oM{V_1}{W}$ and  $n=0,\dots,N-1$, we denote 
\begin{gather}\label{eq:aks-powers-Z}
A_{n+1}\ceq Z^{n}A_1\subset V_{n+1} 
\end{gather} and $P_{A_{n+1}}$ the corresponding orthogonal projection.
We observe from \eqref{eq:ZV-properties02} that $\bigoplus_{j=1}^{N}A_j\subset \oM{V}{W}$.
\begin{lemma}
For $k=1,\dots,N-1$, the projection $P_{A_{k+1}}$ holds 
\begin{gather}
\label{eq:aux-projection}
Z^*P_{A_{k+1}}Z=\abs Z P_{A_k}
\end{gather}
\end{lemma}
\begin{proof}
We check at once that  $A_{k+1}=ZA_k$ and $A_{k+1}^\perp=ZA_k^\perp$, where $A_k^\perp=\oM{P_k\H}{A_k}$. Thus, $Z^*A_{k+1}=\abs ZA_k$ and $Z^*A_{k+1}^\perp=\abs ZA_k^\perp$. 
 Moreover, one easily verify that $Z^*P_{A_{k+1}}Z$ is a projection and  
\begin{align*}
Z^*P_{A_{k+1}}Z\H&={Z_k^*}P_{A_{k+1}}Z_kP_k\H
\\&={Z_k^*}P_{A_{k+1}}Z_kA_k\oplus{Z_k^*}P_{A_{k+1}}Z_kA_k^\perp\\
&={Z_k^*}P_{A_{k+1}}A_{k+1}\oplus{Z_k^*}P_{A_{k+1}}A_{k+1}^\perp=\abs ZA_k\,,
\end{align*} 
which yields \eqref{eq:aux-projection}.
\end{proof}

For convenience of notation, from now on set $\beta_0=0$.
\begin{theorem}\label{th:characterisation-invariant-state}
If $\rho$ is an invariant state supported on $\oM{V}{W}$, then 
\begin{gather}\label{eq:characterisation-invariant-state}
\rho= \tr(|Z|_{1}\rho)\sum_{n=0}^{N-1}e^{\sum_{j=0}^{n}\beta_j}Z^n \tau Z^{*n} 
\end{gather} where $\tau=\frac{1}{\tr(|Z|_{1}\rho)}\abs Z_1\rho$ is a state supported on $\oM{V_1}{W}$.
\end{theorem}
\begin{proof}
By virtue of Theorem~\ref{th:invariant-states}, for $k=1,\dots,N-1$, $\rho$ commutes with $\abs Z_k,P_{k+1}$ and satisfies \eqref{eq:equivalency-rho-absZ}, which used recursively yields 
\begin{align}\label{eq:powers-Z-recursively}
\begin{split}
 P_{k+1}\rho&= e^{\beta_k}ZP_k \rho Z^*=\cdots\\
&=e^{\sum_{j=0}^k\beta_j}Z^kP_1\rho  Z^{*k}=e^{\sum_{j=0}^k\beta_j}Z^k\abs Z_1\rho Z^{*k}\,,
\end{split}
\end{align}
since $ZP_1=Z\abs Z_1$. Observe that  $W\subset \ker\rho$. Hence, one infers from \eqref{eq:ker-rho-inOmega} that $\rho=\abs Z_1\rho+\sum_{k=1}^{N-1}P_{k+1}\rho$, whence by \eqref{eq:powers-Z-recursively}one arrives at \eqref{eq:characterisation-invariant-state}. Clearly, $\tau$ is a state since $\rho$ is. 
\end{proof}
\begin{remark}\label{rm:no-supported-absZ} Theorem~\ref{th:characterisation-invariant-state} implies that no invariant states supported on $\oM{V}{\ran P_1}$ exist. Otherwise, in \eqref{eq:characterisation-invariant-state} $\tau=0$, i.e., $\rho=0$, a contradiction. 
\end{remark}
\begin{corollary}
If an invariant state is supported on $\oM{V}{\ran \abs Z_1}$ then so is on $W$.
\end{corollary}
\begin{proof}
Since $\rho$ is supported on $\oM{V}{\ran \abs Z_1}$, then from Theorem~\ref{th:invariant-states}, it commutes with $\abs Z_1^\perp$ and  due to Proposition~\ref{prop:convex-decomposition}, it decomposes into a convex combination of two invariant sates $\eta,\tau$, supported on $W$ and $\oM{V}{\ran P_1}$,  respectively. Hence, we infer from Remark~\ref{rm:no-supported-absZ} that $\rho=\eta$.
\end{proof}

Let us proceed to prove the converse of Theorem~\ref{th:characterisation-invariant-state}. Any state  decomposes into a convex combination of pure states, which certainly are rank-one projections \cite[Sec.\,2.1.3]{MR2012610}.
 \begin{theorem}\label{th:invariant-states-from-statesV1}
If $\tau$ is a state supported on $\oM{V_1}{W}$, then 
\begin{gather}\label{eq:characterisation-invariant-state-level1}
c \sum_{n=0}^{N-1}e^{\sum_{j=0}^{n}\beta_j}Z^n\tau Z^{*n}
\end{gather}
is an invariant state supported on $\oM{V}{W}$, where $c$ is a normalization constant. Besides, the range of \eqref{eq:characterisation-invariant-state-level1} is $\oM{V}{W}$, when $\ran \tau=\oM{V_1}{W}$.
\end{theorem}
\begin{proof}
Due to linearity, we only need to consider the case when $\tau$ is a one-rank projection $P_{A_1}$ of $\H$ onto a one-dimensional subspace $A_1\subset \oM{V_1}{W}$. It is clear that \eqref{eq:characterisation-invariant-state-level1} is a state. So, for $k=1,\dots,N-1$, by means of \eqref{eq:aks-powers-Z} and \eqref{eq:aux-projection}, $P_{A_{k+1}}=Z\abs Z P_{A_k}Z^*=ZP_{A_k}Z^*$, whence recursively one obtains 
\begin{gather}\label{eq:Zn-PA1-Znstar}
Z^{n}P_{A_1}Z^{*n}=P_{A_{n+1}}\,,\quad n=0,\dots,N-1\,.
\end{gather}
Thus, denoting \eqref{eq:characterisation-invariant-state-level1} by $\hat\tau$, one computes from \eqref{eq:aux-projection} and \eqref{eq:Zn-PA1-Znstar} that
\begin{align*}
Z_k^*\hat\tau Z_k= c e^{\sum_{j=0}^{k}\beta_{j}} Z_k^*P_{A_{k+1}} Z_k =e^{\beta_k}\left(c e^{\sum_{j=0}^{k-1}\beta_j}  \abs Z_kP_{A_k}\right) =e^{\beta_k}\abs Z_k\hat\tau 
\end{align*}
wherefrom $\hat\tau$ commutes with $\abs Z_k$. Besides by \eqref{eq:Zn-PA1-Znstar}, $\hat\tau$ also commutes with $P_{k+1}$ and $\ran \hat\tau\subset\oM VW$. Therefore, $\hat\tau$ is invariant since satisfies all conditions in   Theorem~\ref{th:invariant-states}. If $A_1=\oM{V_1}{W}=\abs ZV_1$,  then item \ref{it:isom-V02} of Lemma~\ref{lem:properties-trasport-operator} yields $A_{k+1}=V_{k+1}$ and  \eqref{eq:ZV-properties02}, \eqref{eq:Zn-PA1-Znstar} imply $\ran \hat\tau=\oM{V}{W}$.\end{proof}

States which cannot be represented as a nontrivial convex combination of two different states are called extremal. Analogously, a state $\rho$ is called \emph{invariant-extremal}, if it is invariant and no nontrivial convex combination of different invariant states exists which coincides with $\rho$. 

\begin{lemma}\label{lem:invariant-extremal-vs-extremal}
An invariant state $\rho$ supported on $\oM{V}{W}$, is invariant-extremal if and only if $\tau= \frac{1}{\tr(\abs Z_1\rho)}\abs Z_1\rho$ is extremal.
\end{lemma}
\begin{proof}
The proof carries out by contraposition. If $\tau$ is no extremal then $\tau=\lambda_1\tau_1+\lambda_2\tau_2$, where $\tau_1,\tau_2$ are states and  $\lambda_1,\lambda_2>0$, with $\lambda_1+\lambda_2=1$. Besides, by virtue of Theorem~\ref{th:characterisation-invariant-state},
\begin{align}\label{eq:decomposition-rho-invariant-states}
\rho=\sum_{l=1}^2\tr (\abs Z_1\rho)\beta_l\lambda_l\left(\frac{1}{\beta_l}
\sum_{n=0}^{N-1}e^{\sum_{j=0}^{n}\beta_j}Z^n\tau_l Z^{*n}\right)\,,
\end{align}
with $\beta_l=\tr(\sum_{n=0}^{N-1}e^{\sum_{j=0}^{n}\beta_j}Z^n\tau_l Z^{*n})$. Note, $\sum_{l=1}^2\tr (\abs Z_1\rho)\beta_l\lambda_l=1$. Hence, Theorem~\ref{th:invariant-states-from-statesV1} and \eqref{eq:decomposition-rho-invariant-states} imply that $\rho$ is decomposed into a convex combination of two invariant states, i.e., it is no 
invariant-extremal. 
Conversely, suppose that $\rho$ is a nontrivial convex combination $\lambda_1\rho_1+\lambda_2\rho_2$, of two invariant states $\rho_1,\rho_2$. Then, 
\begin{align*}
\tau=\frac{\lambda_1\tr(\abs Z_1\rho_1)}{\tr(\abs Z_1\rho)}\left(\frac{\abs Z_1\rho_1}{\tr(\abs Z_1\rho_1)}\right)+\frac{\lambda_2\tr(\abs Z_1\rho_2)}{\tr(\abs Z_1\rho)}\left(\frac{\abs Z_1\rho_2}{\tr(\abs Z_1\rho_2)}\right)\,,
\end{align*}
which is a convex combination of two states. Therefore, $\tau$ is no extremal.\end{proof}

The following assertion is a simple consequence of Theorem~\ref{th:invariant-states-from-statesV1}, Lemma~\ref{lem:invariant-extremal-vs-extremal} and the fact that pure states are extremal.
\begin{corollary}
If $u$ is a unit vector in $\oM{V_1}{W}$,  then
\begin{gather*}
c \left(\sum_{n=0}^{N-1}e^{\sum_{j=0}^{n}\beta_j}Z^n\vk u\vb u Z^{*n}\right)
\end{gather*}
is an invariant-extremal state supported on $\oM{V}{W}$, where $c$ is a normalization constant.
\end{corollary}
\begin{remark}
Clearly, a pure invariant state is invariant-extremal. Particularly, any pure state supported on $W$, is invariant-extremal.\end{remark}
\begin{theorem}
If $\rho$ is an invariant-extremal state supported on $V$, then one of the following conditions is true:
\begin{enumerate}
\item\label{item:extremal-01} There exists a unit vector $u\in W$, such that  $\rho=\vk u\vb u$.
\item\label{item:extremal-02} There exists a vector $u\in \oM{V_1}{W}$ such that  $\dS\rho=\sum_{n=0}^{N-1}e^{\sum_{j=0}^{n}\beta_j}Z^n\vk u\vb u Z^{*n}$.
\end{enumerate}
\end{theorem}
\begin{proof}
From Theorem~\ref{th:invariant-states}, $\rho$ commutes with $\abs Z_1^\perp$ which by Proposition~\ref{prop:convex-decomposition}, must be supported on either $W$ or $\oM{V}{W}$, since otherwise it is no invariant-extremal. Thus,  if $\ran\rho\subset W$, then $\rho$ is a convex combination of pure states supported on $W$, which are  invariant. Hence, $\rho$ is pure, since it is invariant-extremal and it satisfies item \ref{item:extremal-01}.

On the other hand, if $\ran\rho\subset \oM{V}{W}$, then by Theorem~\ref{th:characterisation-invariant-state}, the state $\rho$ holds \eqref{eq:characterisation-invariant-state} and by Lemma~\ref{lem:invariant-extremal-vs-extremal}, $\frac{1}{\tr(\abs Z_1\rho)}\abs Z_1\rho=\vk {\tilde u}\vb {\tilde u}$, where $\tilde u\in\oM{V}{W}$ is a unit vector. Hence, denoting $u=\sqrt{\tr(\abs Z_1\rho)}\,\tilde u$, one gets item \ref{item:extremal-02}.
\end{proof}

Now, we work with states supported on $\Omega^\perp=\Span\{\vk-,\vk+, \varphi_{0_1},\varphi_{0_N}\}$. It is a simple matter to compute that if $\rho$ is a state supported on $\Omega^\perp$, then
\begin{align*}
\cA L(\rho)=&{}\eta_-[P_-,\rho]+\eta_{+}P_+\rho+\cc{\eta}_{+}\rho P_++\eta_{0_1}P_{\varphi_{0_1}}\rho+\cc{\eta}_{0_1}\rho P_{\varphi_{0_1}}+\eta_{0_N}P_{\varphi_{0_N}}\rho\\&+\cc{\eta}_{0_N}\rho P_{\varphi_{0_N}}
+n_1\Gamma_{-,\omega_+}\vk{\varphi_{0_1}}\vb{+}\rho\vk{+}\vb{\varphi_{0_1}}+n_1\Gamma_{+,\omega_+}\vk{+}\vb{\varphi_{0_1}}\rho\vk{\varphi_{0_1}}\vb{+}\\
&+\Gamma_{-,\omega_-}\vk{-}\vb{\varphi_{0_N}}\rho\vk{\varphi_{0_N}}\vb{-}+\Gamma_{-,\omega_1}Z_1\rho Z_1^*+\Gamma_{+,\omega_{N-1}}Z_{N-1}^*\rho Z_{N-1}\,, 
\end{align*}
where $P_a$ is the orthogonal of $\H$ projection onto $a\in\{\vk-,\vk+, \varphi_{0_1},\varphi_{0_N}\}$ and 
\begin{align*}
\eta_{-}&=i\gamma_{+,\omega_-}&
\eta_{0_1}&=i(n_1\gamma_{+,\omega_+}-\gamma_{-,\omega_1})-\frac{n_1\Gamma_{+,\omega_+}+\Gamma_{-,\omega_1}}2\\
\eta_{+}&=-\frac{n_1}{2}(i2\gamma_{-,\omega_+}+\Gamma_{-,\omega_+})&
\eta_{0_N}&=i(\gamma_{+,\omega_{N-1}}-\gamma_{-,\omega_-})-\frac{\Gamma_{+,\omega_{N-1}}+\Gamma_{-,\omega_-}}2
\end{align*}
\begin{remark} It is easily seen  that $\cA L(P_-)=0$, i.e., $P_-$ is an invariant state.
\end{remark}
\begin{theorem}\label{omega-perp}
If $\rho$ is an invariant state supported on $\Omega^\perp$ then $\rho=P_-$.
\end{theorem}
\begin{proof}
It suffices to prove that $\rho P_a=0$, for $a\in\{\vk+,\varphi_{0_1},\varphi_{0_N}\}$. Note that $0=P_a\cA L(\rho)P_{\varphi_{0_N}}=(\eta_{a}+\cc{\eta}_{0_N})P_a\rho P_{\varphi_{0_N}}$, which implies $P_a\rho P_{\varphi_{0_N}}=0$, for $a\in\{\vk-,\vk+,\varphi_{0_1},\varphi_{0_N}\}$. Then, 
\begin{eqnarray}\label{eq:rho-varphiN}
\rho P_{\varphi_{0_N}}&=\dS\sum_{\substack{a\in \\ \{\vk-,\vk+,\varphi_{0_1}, \varphi_{0_N}\}}}P_a\rho P_{\varphi_{0_N}}=&0\,.
\end{eqnarray}
Beside, for $a\in\{\vk-,\varphi_{0_1}\}$, one has  $0=P_a\cA L(\rho)P_{+}=(\eta_{a}+\cc{\eta}_{+})P_a\rho P_{+}$, which yields $P_a\rho P_{+}=0$ and 
\begin{eqnarray*}
[\rho,P_+]&=\dS\sum_{\substack{a\in \\ \{\vk-,\vk+,\varphi_{0_1}\}}}P_a(\rho P_+-P_+\rho)\sum_{\substack{a\in \\ \{\vk-,\vk+,\varphi_{0_1}\}}}P_a=&0\,,
\end{eqnarray*}
whence it follows that $\rho$ commutes with $P_+$. Moreover,
\begin{align}\label{eq:mas-mas}
\begin{split}
0&=P_+\cA L(\rho)P_{+}=-n_1\Gamma_{-,\omega_+}\rho P_++n_1\Gamma_{+,\omega_+}\vk{+}\vb{\varphi_{0_1}}\rho\vk{\varphi_{0_1}}\vb{+}\,;\\[3mm]
0&=\vk{+}\vb{\varphi_{0_1}}\cA L(\rho)\vk{\varphi_{0_1}}\vb{+}\\&=n_1\Gamma_{-,\omega_+}\rho P_+-(n_1\Gamma_{+,\omega_+}+\Gamma_{-,\omega_1})\vk{+}\vb{\varphi_{0_1}}\rho\vk{\varphi_{0_1}}\vb{+}\,.
\end{split}
\end{align}
Thus, \eqref{eq:mas-mas} yields $\rho P_+=0$ and $\vk{+}\vb{\varphi_{0_1}}\rho\vk{\varphi_{0_1}}\vb{+}=0$. This also produces $P_{\varphi_{0_1}}\rho P_{\varphi_{0_1}}=0$. Note that $0=P_-\cA L(\rho)P_{\varphi_{0_1}}$ implies $P_-\rho P_{\varphi_{0_1}}=0$. Hence, analogously to \eqref{eq:rho-varphiN}, we conclude that $\rho P_{\varphi_{0_1}}=0$.
\end{proof}
\begin{remark}
On account of Proposition~\ref{prop:omega-vs-V}, Theorem~\ref{th:invariant-states-from-statesV1} and Remark~\ref{rm:W-in-V}, the biggest support of an invariant state supported on $\Omega$ is $V$. On the other hand according to Theorem~\ref{omega-perp} the only invariant state supported on $\Omega^\bot$ is $P_-$. 
\end{remark}

In view that it remains unknown whether invariant states with nontrivial support on both $\Omega$ and $\Omega^\perp$ exist, we finalize with the following conjecture.
Recall that the fast recurrent subspace is the biggest support of invariant states (cf. \cite{MR4107240}).\\[3mm]
\begin{conjecture}
The fast recurrent subspace of $\cA L$  is $V\oplus\C\vk{-}$.
\end{conjecture}
This issue exceeds the scope of this work and will be tackled in a forthcoming paper.
\begin{acknowledgments}
The financial support from CONACYT-Mexico (Grant 221873), PRODEP Red de An\'alisis Italia-UAM M\'exico and  posdoctoral fellowship 511-6/2019.-10951, is gratefully acknowledged.  
\end{acknowledgments}

\def\cprime{$'$} \def\lfhook#1{\setbox0=\hbox{#1}{\ooalign{\hidewidth
  \lower1.5ex\hbox{'}\hidewidth\crcr\unhbox0}}} \def\cprime{$'$}
  \def\cprime{$'$} \def\cprime{$'$} \def\cprime{$'$} \def\cprime{$'$}
  \def\cprime{$'$} \def\cprime{$'$}
\providecommand{\bysame}{\leavevmode\hbox to3em{\hrulefill}\thinspace}
\providecommand{\MR}{\relax\ifhmode\unskip\space\fi MR }
\providecommand{\MRhref}[2]{%
  \href{http://www.ams.org/mathscinet-getitem?mr=#1}{#2}
}
\providecommand{\href}[2]{#2}

\end{document}